\documentclass[runningheads]{llncs}
\usepackage{algorithm}
\usepackage[noend]{algpseudocode}
\usepackage{amsmath}
\usepackage{amssymb}
\usepackage{gensymb}
\usepackage{graphicx}
\usepackage{mathtools}
\usepackage[section]{placeins}
\usepackage{subcaption}
\DeclarePairedDelimiter{\ceil}{\lceil}{\rceil}

\makeatletter
\def\BState{\State\hskip-\ALG@thistlm}
\floatname{algorithm}{Procedure}
\makeatother

\begin{document}

\title{Centralised Connectivity-Preserving Transformations for Programmable Matter: A Minimal Seed Approach}

\titlerunning{Connectivity-Preserving Transformations for Programmable Matter}

\author{Matthew Connor \and
Othon Michail \and
Igor Potapov}

\authorrunning{M. Connor et al.}

\institute{Department of Computer Science, University of Liverpool, UK\\
\email{M.Connor3@liverpool.ac.uk, Othon.Michail@liverpool.ac.uk, potapov@liverpool.ac.uk}\\
}

\maketitle

\begin{abstract}
We study a model of programmable matter systems consisting of $n$ devices lying on a 2-dimensional square grid which are able to perform the minimal mechanical operation of rotating around each other. The goal is to transform an initial shape A into a target shape B. We investigate the class of shapes which can be constructed in such a scenario under the additional constraint of maintaining global connectivity at all times. We focus on the scenario of transforming nice shapes, a class of shapes consisting of a central line $L$ where for all nodes $u$ in $S$ either $u \in L$ or $u$ is connected to $L$ by a line of nodes perpendicular to $L$. We prove that by introducing a minimal 3-node seed it is possible for the canonical shape of a line of $n$ nodes to be transformed into a nice shape of $n-1$ nodes. We use this to show that a 4-node seed enables the transformation of nice shapes of size $n$ into any other nice shape of size $n$ in $O(n^2)$ time. We leave as an open problem the expansion of the class of shapes which can be constructed using such a seed to include those derived from nice shapes.\\

\end{abstract}

\section{Introduction}\label{sec1}

Programmable matter refers to matter which can change its physical properties algorithmically.
This means that the change is the result following the procedure of an underlying program.
The implementation of the program can either be a system level external centralised algorithm or an internal decentralised algorithm executed by the material itself.
The model for such systems can be further refined to specify properties that are relevant to real-world applications, for example connectivity, colour \cite{CLSL11} and other physical properties.
The result of this is the development of programmable materials such as self-assembling DNA molecules \cite{D12,R06}.
In addition, systems which rely on large collectives of identical robots have been developed, for example the Kilobot system \cite{RCN14} and the Robot Pebbles system \cite{GKR10}.
Another interesting implementation is Millimotein \cite{KCLO12}, a system where programmable matter folds itself into arbitrary 3D shapes.
The CATOMS system \cite{TPB19,TPB20} is a further implementation which constructs 3D shapes by first creating a ``scaffolding structure" as a basis for construction.
It is expected that applications in further domains such as molecular computers and self-repairing machines may become apparent in the long-term.

As the development of these systems continues, it becomes increasingly necessary to develop theoretical models which are capable of describing and explaining the emergent properties, possibilities and limitations of such systems in an abstract and fundamental manner. To this end, models have been developed for programmable matter. For example, algorithmic self-assembly \cite{D12} focuses on programming molecules like DNA to grow in a controllable way, and the Abstract Tile Assembly Model \cite{RW00,W98}, as well as the nubot model \cite{WCGD13}, have both been developed for this area. Network Constructors \cite{MS16} uses the Population Protocol model \cite{AADFP06} based on a population of finite-automata interacting randomly as the basis for a new model where the automata are able to create networks by forming connections with each other. The latter model is formally equivalent to a restricted version of chemical reaction networks, which ''are widely used to describe information processing occurring in natural cellular regulatory networks'' \cite{SCWB08,D13}.
Finally there is extensive research into the amoebot model \cite{DDGR14,DDGP18,DGRS15,DGRS16}, where finite automata on a triangle lattice follows a distributed algorithm to achieve a desired goal.

Recent progress in this direction has been made in a previous paper \cite{AMP20}, covering questions related to a specific model of programmable matter where nodes exist in the form of a shape on a 2D grid and are capable of performing two specific movements: rotation around each other and sliding a node across two other nodes. They presented 3 problems: transformations with only rotations (Rot-Transformability), transformations with rotations with the restriction that shapes must always remain connected (RotC-Transformability) and transformations with both rotation and sliding movements (RS-Transformability). For Rot-Transformability they prove universal transformation between any pair of colour-consistent shapes which are not blocked, however they leave universal RotC-Transformability as an open problem. Such transformations are highly desirable due to the large numbers of programmable matter systems which rely on the preservation of connectivity. Progress in a very similar direction was made in another paper \cite{AADD21}, which used a similar model but allowed for a greater range of movement, for example ``leapfrog" and ``monkey" movements. They accomplished universal transformation in $O(n^2)$ movements using a ``bridging" procedure which added up to 5 nodes during the procedure as necessary in a manner similar to the seed idea from the previous paper.

\section{Contribution}\label{sec2}

We investigate which families of connected shapes can be transformed into each other via rotation movements without breaking connectivity.

We consider the case of programmable matter on a 2D grid which is only capable of performing rotation movements, defined as the $90\degree$ rotation of a node $a$ around one of the two vertices of the edge it shares with a neighbouring edge-adjacent node $b$, so long as the goal and intermediate cells are empty. All nodes must be \emph{edge connected}, meaning that at every time step there must be a path from any arbitrary node to any other node crossing only spaces occupied by nodes via edges. Our algorithms are \emph{centralised}, using external procedures to transform shapes, and therefore focusing on the questions of the feasibility and complexity of the transformations.

We assume the existence of a \emph{seed}, a group of nodes in an arbitrary shape $S$ which are placed in empty cells neighbouring a shape $A$ to create a new connected shape which is the unification of $S$ and $A$.
Seeds allow shapes which are blocked or incapable of meaningful movement to perform otherwise impossible transformations.
The use of seeds was established in a previous work \cite{MSS19}, and more recently shown to enable universal reconfiguration in the context of connectivity preserving transformations \cite{AADD21}, however to our knowledge there has been no attempt to investigate this problem using a seed which is a connected shape fully introduced before the transformation is initiated.

We first study blocked shapes, where our goal is to define the class of shapes which are \emph{blocked}, or incapable of moving any node without a seed.
We show that shapes of this class consist of nodes which are surrounded by diagonal lines in the shape of a rhombus, or overlapping rhombuses which may be connected by lines.
We then investigate the transformation of nice shapes.
A \emph{Nice Shape} (defined in \cite{AMP20}) is a shape $S$ which has a central line $L$ where for all nodes $u$ in $S$ either $u \in L$ or $u$ is connected to $L$ by a line of nodes perpendicular to $L$.
We provide a lower bound of $Omega(n^2)$ for transforming a line of $n$ nodes into a nice shape.
We show that it is possible to transform such a line into a nice shape of $n-1$ nodes using a 3-node seed in $O(n^2)$ time.
We then demonstrate that it is possible to transform nice shapes of size $n$ into other nice shapes of size $n$ by using the canonical shape of a line and a 4-node seed in $O(n^2)$ time.
We provide an algorithm to implement this transformation and give time bounds for it.
We then provide further directions for research.

In Section \ref{sec3}, we formally define the model of connectivity-preserving programmable matter used in this paper.
In Section \ref{sec4} we give our lower bounds.
In Section \ref{sec5} we provide our algorithm  for the construction of nice shapes where the colour of nodes added to each side of the line always alternates, then generalise first to all nice shapes and second to the class of shapes made up of nice shapes.
In Section \ref{sec6} we conclude and give directions for potential future research.

\section{Model}\label{sec3}
The programmable matter systems considered in this paper operate on a 2D square grid, with each cell being uniquely referred to by its $y\geq0$ and $x\geq0$ coordinates. Such a system consists of a set $S$ of $n$ modules, called nodes throughout. Each node may be viewed as a spherical module fitting inside a cell. At any given time, each node $u \in V$ occupies a cell in the grid $o(u) = (o_y(u), o_x(u)) = (i, j)$ (where $i$ corresponds to a row and $j$ to a column of the grid) and no two nodes may occupy the same cell. At any given time $t$, the positioning of nodes on the grid defines an undirected neighboring relation $E(t) \subset S \times S$, where $\{u, v\} \in E$ iff $o_y(u) = o_y(v)$ and $|o_x(u) - o_x(v)| = 1$ or $o_x(u) = o_x(v)$ and $|o_y(u) - o_y(v)| = 1$, that is, if $u$ and $v$ are either horizontal or vertical neighbors on the grid, respectively. A more informative and convenient way to define the system at any time $t$ is the mapping $P_t:\mathbb{N}_{\geq 0} \times \mathbb{N}_{\geq 0} \rightarrow \{0, 1\}$.where $P_t(i, j) = 1$ iff cell (i, j) is occupied by a node. At any given time $t$, $P_{t}^{-1}(1)$ defines a shape. Such a shape is called \emph{connected} if $(V, E(t))$ defines a connected graph. 

In general, shapes can \emph{transform} to other shapes via a sequence of one or more movements of individual nodes. We consider only one type of movement: rotation. In this movement, a single node moves relative to one or more neighboring nodes. A single rotation movement of a node $u$ is a 90° rotation of $u$ around one of its neighbors. Let $(i, j)$ be the current position of $u$ and let its neighbor be $v$ occupying the cell $(i-1, j)$ (i.e., lying below $u$). Then $u$ can rotate 90° clockwise (counterclockwise) around $v$ iff the cells $(i, j + 1)$ and $(i-1, j + 1)$ ($(i, j-1)$ and $(i-1, j-1)$, respectively) are both empty. By rotating the whole system 90°, 180°, and 270°, all possible rotation movements are defined analogously.

Let $A$ and $B$ be two connected shapes. We say that $A$ transforms to $B$ via a rotation $r$, denoted $A \stackrel{r}{\rightarrow} B$, if there is a node $u$ in $A$ such that if $u$ applies $r$, then the shape resulting after the rotation is $B$. We say that $A$ transforms in one step to $B$ (or that $B$ is reachable in one step from $A$), denoted $A \rightarrow B$, if $A \stackrel{r}{\rightarrow} B$ for some rotation $r$. We say that $A$ transforms to $B$ (or that $B$ is reachable from $A$) and write $A \rightsquigarrow B$, if there is a sequence of shapes $A = C_0, C_1, . . . , C_t = B$, such that $C_i \rightarrow C_{i+1}$ for all $0 \leq i < t$. Rotation is a reversible movement, a fact that we use in our results.

A line is a connected shape where every node lies on the same column or the same row.
A nice shape $N$ is defined as a shape which has a central line $L$ where for all nodes $u$ either $u \in L$ or $u$ is connected to $L$ by a line of nodes perpendicular to $L$.

Consider a black and red checkered colouring of the 2D grid, like that of a chessboard.
Then any shape $S$ consists of $b(S)$ nodes which lie on black cells and $r(S)$ nodes which lie on red cells.
Two shapes $A$ and $B$ are \emph{colour consistent} if $b(A) = b(B)$ and $r(A) = r(B)$.
If $S$ is not a nice shape and $S = A \cup B$ where $A$ is a nice shape, we call $B$ the \emph{waste} of the shape $S$.
A \emph{configuration} of a shape is an arrangement of the nodes of the shape on a 2D grid where each node is uniquely identifiable.

\section{Infeasible Transformations and the Time Lower Bound}\label{sec4}

In this section, we cover a series of transformations which are infeasible, meaning that they rely on the ability to move $O(n)$ nodes but exist in a scenario where moving at most $O(1)$ is possible.
We first define the class of shapes which are \emph{blocked}, meaning there is no potential movement available for any node.
We then define the class of \emph{k-blocked shapes}, where the set of potential transformations has at most $k$ configurations before any configuration is repeated.
We show that it is necessary for a seed to have at least 3 nodes if it is to be connected and to enable the movement of more than 5 nodes in a horizontal line.
Finally, we provide a lower bound of $\Omega(n^2)$ movements for the problem of transforming a line into a nice shape.

A node $w$ is an \emph{interior} node if for each of the cells $x$ edge-adjacent to $w$ either there is a node occupying $x$ or there are two nodes $y$ and $z$ such that $y$ and $z$ are edge-adjacent to $x$ and vertex-adjacent to $w$.
A node is an \emph{exterior} node if it is not an interior node.

\begin{theorem}
An arbitrary shape \emph{A} is blocked if and only if there is only 1 node or every exterior node has no edge connections to any other exterior node.
\end{theorem}

\begin{proof}
A shape with one node is trivially blocked because there is nothing for it to rotate around.

Otherwise, a shape consists of interior nodes connected to each other with the possibility of one-node gaps, surrounded by exterior nodes which form diagonal lines due to the edge-adjacency restriction.

Interior nodes are blocked by the nodes that surround them, either because the grid space is filled by an edge-adjacent node or the two vertex-adjacent nodes block the rotation movement.

Exterior nodes can only rotate around nodes which are edge-connected, which must be interior nodes. The nodes which surround an interior node, whether edge or vertex connected, always block an exterior node from moving, regardless of whether they are interior or exterior nodes themselves.

Conversely, if there is an exterior node which is edge-connected to an exterior node, the exterior node can rotate into the empty space which it provides.
\qed
\end{proof}

This creates a shape which is similar to the trapezoid which is known to be blocked, however the key difference is that there are not necessarily four sides. The resulting shapes therefore resemble multiple trapezoids which overlap each other. For examples, see Figure 1.
In addition, with the additional condition of connectivity preservation, it is possible for these trapezoids to be connected by straight lines resembling a geometric cactus form of a cactus graph with rhombuses instead of cycles.
We prove this below.

\begin{theorem}
An arbitrary blocked shape \emph{A} is also blocked under the condition of connectivity preservation if there are lines which end in blocked shapes.
\end{theorem}

\begin{proof}
By Theorem 1, each of the blocked shapes is incapable of movement.
Lines are capable of movement, however if any node except the end nodes moves the resulting shape will not be connected.
In addition, the nodes within a line do not enable the movement of nodes in a blocked shape as the only node which is capable of new movement is responsible for maintaining connectivity between the blocked shape and the line.
Therefore, so long as the end nodes of a line are blocked as they form part of the blocked shape, the nodes of such a shape cannot be moved without breaking connectivity.
\qed
\end{proof}

\begin{figure}
\centering
\begin{subfigure}{.4\textwidth}
  \centering
  \includegraphics[width=.8\linewidth]{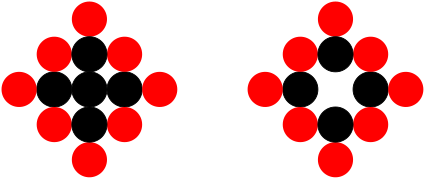}
  \label{fig:sub1}
\end{subfigure}
\begin{subfigure}{.4\textwidth}
  \centering
  \includegraphics[width=.8\linewidth]{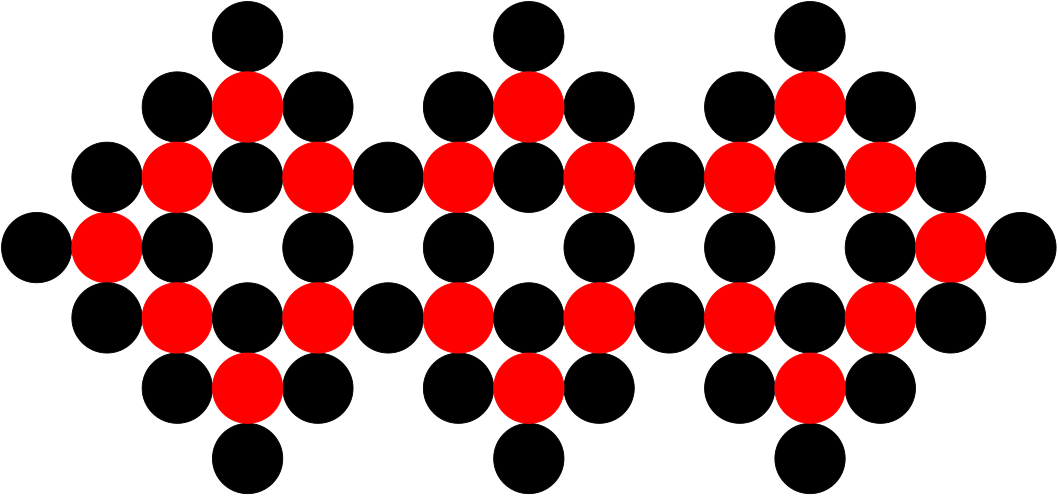}
  \label{fig:sub2}
\end{subfigure}
\begin{subfigure}{.5\textwidth}
  \centering
  \includegraphics[width=.8\linewidth]{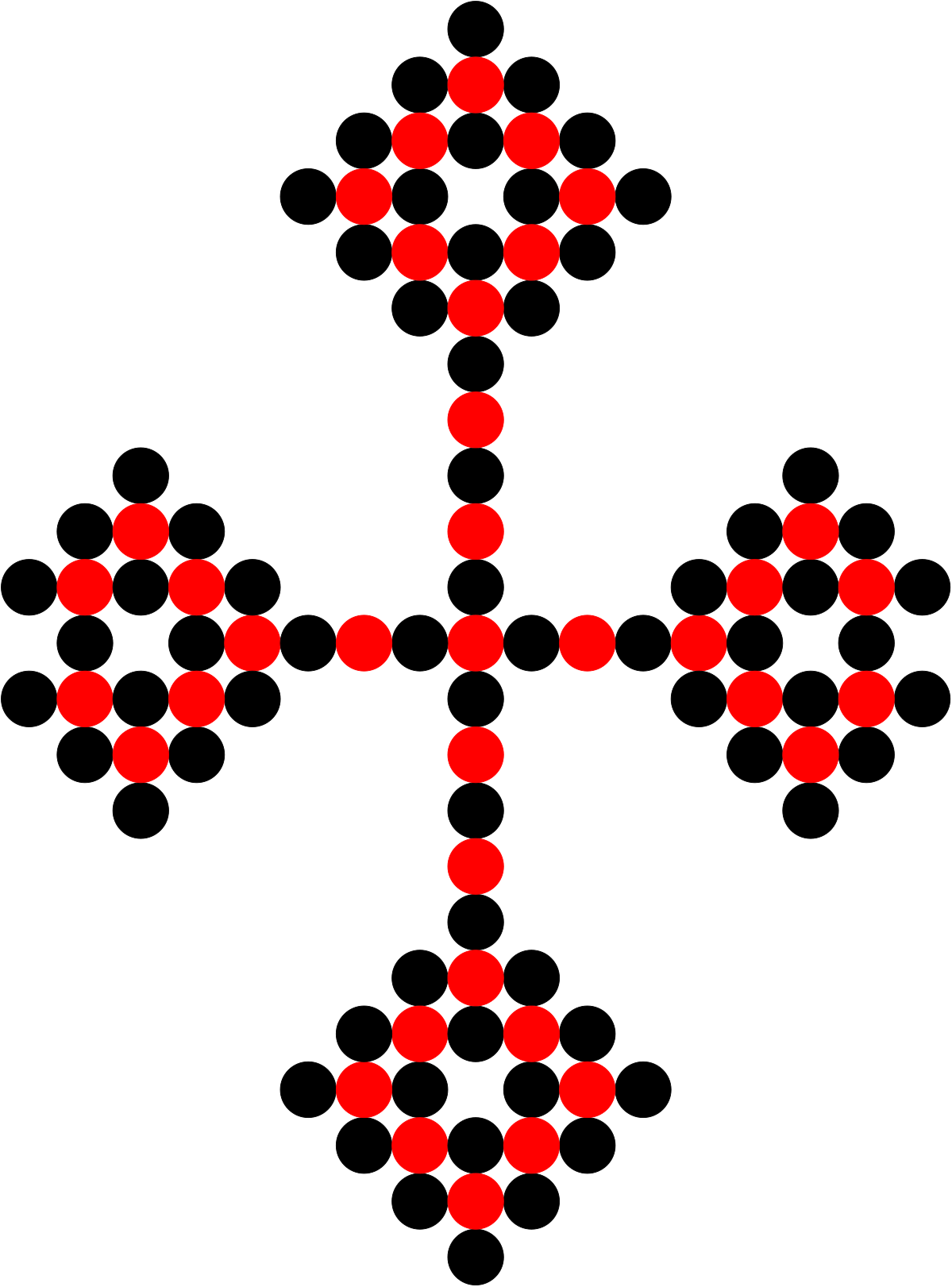}
  \label{fig:sub4}
\end{subfigure}
\caption{Examples of blocked shapes.}
\label{fig:test}
\end{figure}

\FloatBarrier

\subsection{K-Blocked Shapes}

A shape is \emph{K-blocked} if at most $K$ moves can be made before a configuration is repeated.

\begin{lemma}
An arbitrary shape \emph{A} is $K$-blocked if the 2D grid is finite or there is at least one node which cannot move in any configuration of $A$.
\end{lemma}

\begin{proof}
If grid space is finite, then for $n$ nodes in $m > n$ spaces we have a fixed number of assignments of $n$ to $m$, making repetition necessary at some point.

If one node cannot move in any configuration, then it acts an anchor limiting the area a shape can occupy, thanks to the edge connectivity requirement.
This is equivalent to limiting the 2D grid space to Manhattan distance $n$ from a given anchor.
\qed
\end{proof}

\subsection{Time and Seed Lower Bounds for Line Transformations}

We now give a lower bound on the running time of any strategy which transforms a line into a nice shape.

\begin{lemma}
There exists a nice shape such that any strategy which transforms a line of $n$ nodes into the nice shape requires $\Omega(n^2)$ time steps in the worst case.
\end{lemma}

\begin{proof}
Our goal is to transform the line of length $n$ into a nice shape with two lines of length $n/2$, one horizontal line and one vertical line above and perpendicular to the node in the center of the horizontal line.

Let $c$ be the node in the line which the vertical line will be constructed above.

To avoid breaking connectivity, it is necessary for $M$ to transfer nodes from the ends of the line to the space above $c$.
Each of these nodes must perform $\ceil{n/2}$ movements assisted by $M$.
While the distance to the $c$ grows shorter with each node transferred, the line above $c$ grows longer.
Therefore, given that $\ceil{n/2}$ nodes must move towards and onto the vertical line, the total number of movements $m$ is given by $\ceil{n/2} \cdot \ceil{n/2} = \Omega(n^2)$.

\qed
\end{proof}

We define a \emph{connected} seed to be a seed which is a connected shape by itself.
We next show that a connected seed of size $s < 3$ can only move a constant number of nodes (5).
Note that if the seed is disconnected a 2-node seed is able to enable non-trivial movement by placing the nodes of the seed in positions $(2, 1)$ and $(n - 3, 1)$ and working with the nodes at both ends of the line at the same time.
The position of the seed can also be symmetrical i.e. above $(n - 1, 0)$ to $(n - 3, 0)$ or even below the line so long as the destination of the pairs is also mirrored.

\begin{lemma}
Any line of nodes $S$ of length $n$ can move at most five nodes from the line with any k-seed of size $k < 3$ nodes.
\end{lemma}

\begin{proof}
A line without seeds, with the connectivity preserving condition and with only rotation movements cannot do anything other than rotate the two nodes at each end point.
With a one node seed, the only possible action is for the node to be positioned in the cell $(2, 1)$ (or any equivalent symmetrical position) and rotate the end node at $(0, 0)$
to $(1, 1)$ to form a pair. This is equivalent to having a two node seed on a line of length $n - 1$.
With a two node seed, it can only interact with an end node and with each node in the positions $(0, 1), (1, 1)$ or $(1, 1), (2, 1)$ (or any symmetrical position).
In the former case, the end node can only rotate around the node in $(0, 1)$ because it depends on it to maintain connectivity.
In the latter case, the end node can rotate to $(0, 1)$. This allows the node in $(1, 0)$ and the node next to it (i.e. in $(2, 0)$) to rotate.
However, they cannot move much without breaking connectivity thanks to a reliance on the nodes in $(3, 0)$ and $(3, 1)$ for connectivity which restricts movement.

Therefore, if we start with a one node seed, form a two node pair, rotate the node in $(0, 1)$ to $(2, 1)$, move the two nodes in $(1, 0)$ and $(2, 0)$ and the node at the other end of the line, we have exhausted all possibilities to maximise the number of moving nodes without using a seed of size $k \geq 3$.

\qed
\end{proof}

\section{Transformation for Nice Shapes}\label{sec5}

In this section, we investigate the possibilities related to the transformation of shapes which are connectivity preserving.
We focus on the problem of converting a nice shape of $O(n)$ nodes into any other nice shape of $O(n)$ nodes. 
We do this by showing we can transform the canonical shape of a line with $O(n)$ nodes into any nice shape.
Due to reversibility, it follows that any nice shape can be transformed into such a line, and then into another nice shape.
More specifically, we first provide a solution for the variant of this problem (which we call $M$) where all the lines perpendicular to a central line $L$ in the nice shape are such that the node at the end of each line is the opposite colour to the node at the end of its nearest neighbouring lines. We then prove that slight modifications to the method of construction allow for the class of all nice shapes to be constructed.
Our methods construct a shape which is a union of a nice shape with constant waste $O(1)$.

We start with a shape $S$ which is a line of length $n$ occupying the cells $(0, 0)$ to $(0, n - 1)$.
We add a connected 3-node seed to the line as this is the minimum size which allows us to move more than 5 nodes without breaking connectivity.
The result is depicted in Figure 2.
It is possible for our results to apply to a disconnected 2-node seed with a slightly modified procedure but with higher waste.
We place the seed in a specific position as the connected 3-node seed is incapable of movement.
We sketch the line to nice shape proof in the following subsection.

\begin{figure}[h]
\includegraphics[width=0.8\textwidth]{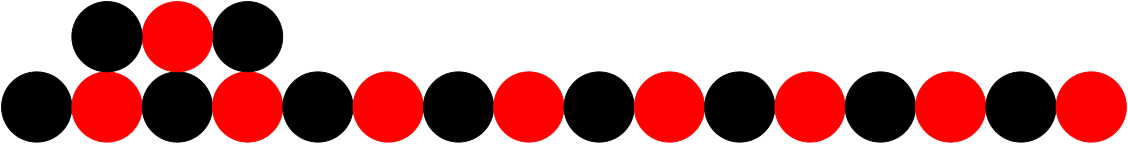}
\caption{The line with the seed attached.} \label{fig1}
\end{figure}

\FloatBarrier

\subsection{Line to Nice Shape}

Our first result is the following theorem:

\setcounter{theorem}{2}
\begin{theorem}
A line of length $n$ can be transformed to any given nice shape in the class $M_{n-1}$ using a 3-node seed in $O(n^2)$ time.
\end{theorem}
\setcounter{theorem}{2}

To solve this problem, we follow a strategy of having nodes rotate onto the horizontal line with the help of the 3 node seed and then constructing lines perpendicular to the horizontal line using the nodes.
Additionally, we move 4 nodes below the line and on the opposite side to the seed.
These nodes can then replicate the behaviour of the seed on the other side of the line, allowing for construction to occur below as well as above the line.
Because their behaviour is the same, we refer to the seed and the group of nodes on the other side of the line as \emph{builders}.
As a result, the horizontal line becomes the central line $L$ of the nice shape, and the vertical lines become the lines of nodes perpendicular to $L$.
Finally, some nodes which aid construction and cannot be incorporated into the final shape are discarded as waste.

To prove that this is possible, we define three algorithmic procedures.
The first procedure, \emph{RaiseNodes}, allows a builder to move two nodes at a time from the horizontal line. These nodes combine with the builder to form a 5 node cluster.
This cluster can be broken if necessary into a 3 node line and a 2 node line, allowing the 2 node line to move by having each node rotate around each other.
The second procedure, \emph{MirrorSeed}, is the procedure for creating the second builder below the horizontal line. It accomplishes this by moving two of the 2 node lines to the end of the horizontal and then rotating nodes in such a way that the four nodes are ``pushed" through the horizontal and to the other side.
The final procedure, \emph{DepositNodes}, collects nodes from the horizontal line and deposits them in any reachable location. We will show that the set of reachable locations is such that depositing nodes in a way which creates any nice shape is possible.

As a result, we end up with nice shapes where the central column $L$ corresponds to what is left of the original horizontal line.
However, the resulting nice shapes have only even lines. This is due to the construction procedures, which place two nodes at a time.
We therefore provide additional movements that allow us to expand the set of nice shapes which can be constructed to include all nice shapes.
We perform this for a special case and then generalise to drop this assumption and get any nice shape.

\subsection{RaiseNodes}

We use a 3 node seed in the cells $(1, 1)$,  $(2, 1)$,  $(3, 1)$ for our operations as, by Lemma 3, a two node seed is incapable of helping nodes to move.

We call the first operation \emph{RaiseNodes}. For this operation we use the 3 node seed to move nodes from the horizontal line such that they are on top of the horizontal line as a pair.
In the process, the 3 node seed moves along the horizontal line such that each node moves from its original position $(x, 1)$ to $(x + 1, 1)$.
The result can also be interpreted as a shape consisting of 5 nodes, which we refer to as a 5-node seed.

Moving the pair of nodes once they are on the line is a trivial process.
Each node rotates around the other node, alternating their relative positions within the two node shape.
As a result, the process can be repeated so long as the pair of nodes on the line can be moved out of the way by rotating around each other to create space.

The following lemma shows that these operations are possible.

\begin{lemma}
Using a 3 node seed in the cells $(1, 1)$ to $(3, 1)$, it is possible to move 2 nodes from the line such that the 3-node seed is converted into a 5-node seed.
\end{lemma}

\begin{proof}
The seed can only be placed in the cells specified as a three node line is incapable of translating to another position.

First, the leftmost node of the three node line must rotate above the middle node.
It is then possible for the first and second nodes in the line $S$, $S_0$ and $S_1$ to rotate around each other along the seed. Only one node ($S_1$) can move around the seed clockwise, reaching the space above the third seed node.
By moving these two nodes one space right (rotating around each other) space is created for the top seed node and first line node to both rotate clockwise. The top seed node and third seed node both rotate clockwise, recreating the shape above the line after the $S_0$ and $S_1$ were initially rotated but one cell to the right. See the figures for a diagram of the transformation.

There is now a five node group above the line. Two of nodes move right to the end, with the other three now occupying the spaces $(3, 1)$ to $(5, 1)$, two spaces to right of the initial position.
\qed
\end{proof}

\begin{figure}
\centering
\begin{subfigure}{.4\textwidth}
  \centering
  \includegraphics[width=.8\linewidth]{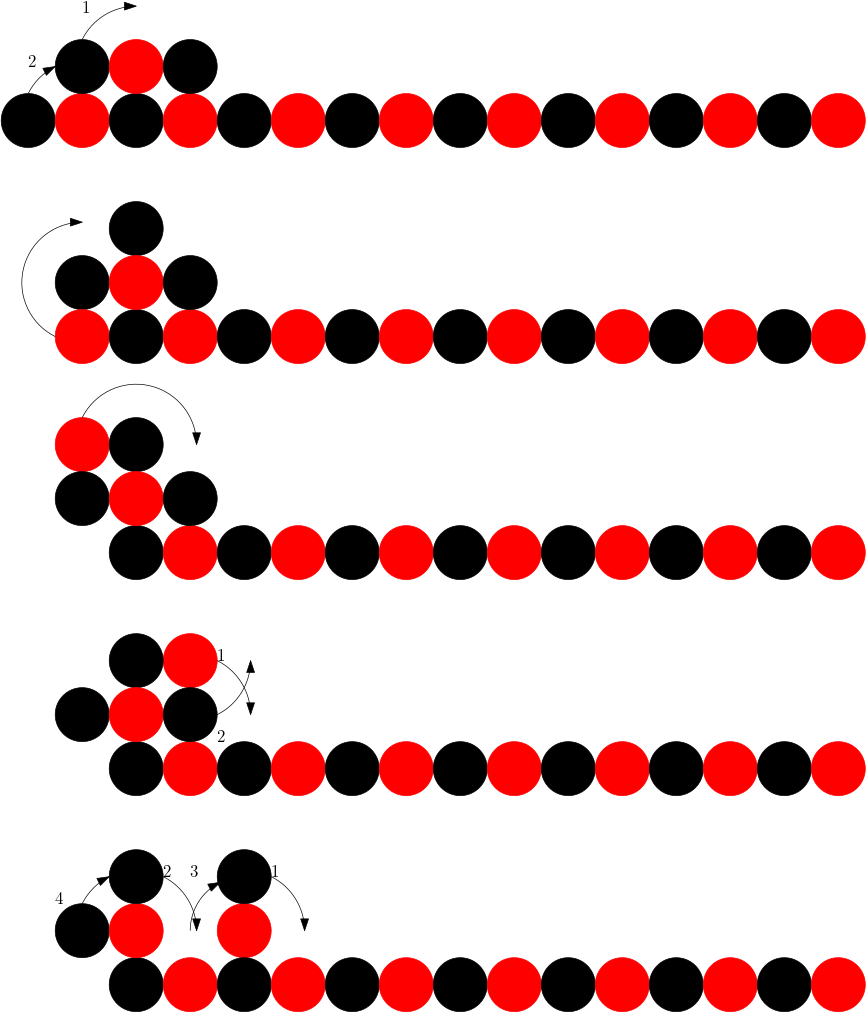}
  \label{fig:sub1}
\end{subfigure}
\begin{subfigure}{.4\textwidth}
  \centering
  \includegraphics[width=.8\linewidth]{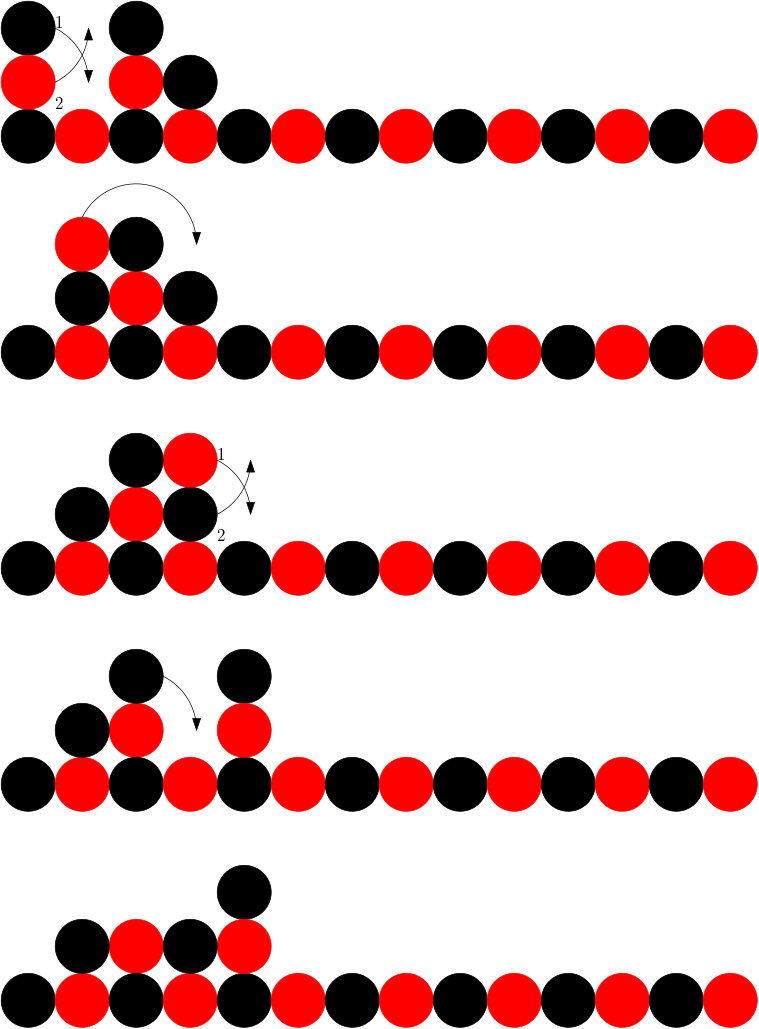}
  \label{fig:sub2}
\end{subfigure}
\caption{Raising nodes from the line.}
\label{fig:test}
\end{figure}

\FloatBarrier

\subsection{MirrorSeed}

We now use RaiseNodes for our next operation, \emph{MirrorSeed}, to place four nodes at the opposite side of the line (i.e. $(n - 4, 1)$ to $(n - 1, 1)$) and then push them through and below the line, creating a four node mirror of our original seed in the cells $(n - 4, -1)$ to $(n - 1, -1)$.
Having a mirror of the original seed allows us to perform construction operations on the bottom of the horizontal line. We do this in 3 steps:
raise four nodes using RaiseNodes twice, position the four nodes at the end of the line and rotate the nodes and those at the end of the line such that the four nodes move through (not around) the line and to the other side.

\begin{lemma}
Using a 3 node seed in the cells $(1, 1)$ to $(3, 1)$, above a line $L$ of length $n$ it is possible to create a 4-node line in the cells immediately below the nodes $(n - 5, 0)$ to $(n - 2, 0)$.
\end{lemma}

\begin{proof}
We first move the 4 leftmost nodes in $S$, $S_0$ to $S_3$ to the top of the line.
We do this by following the procedure depicted in figures $2$ to $4$ to raise $S_0$ and $S_1$, and then repeat the procedure a second time with the next two nodes $S_2$ and $S_3$.
We now have 4 nodes a square above the end of the line. By rotating them around each other in pairs we can place them in the cells $(n - 4, 1)$ to $(n-1, 1)$.
We can then ``push" the nodes to the other side of $S$ by follow the procedure depicted in Figure 4. The result is four nodes in the cells $(n - 4, -1)$ to $(n-1, -1)$
\qed
\end{proof}

\begin{figure}
\centering
\begin{subfigure}{.4\textwidth}
  \centering
  \includegraphics[width=.8\linewidth]{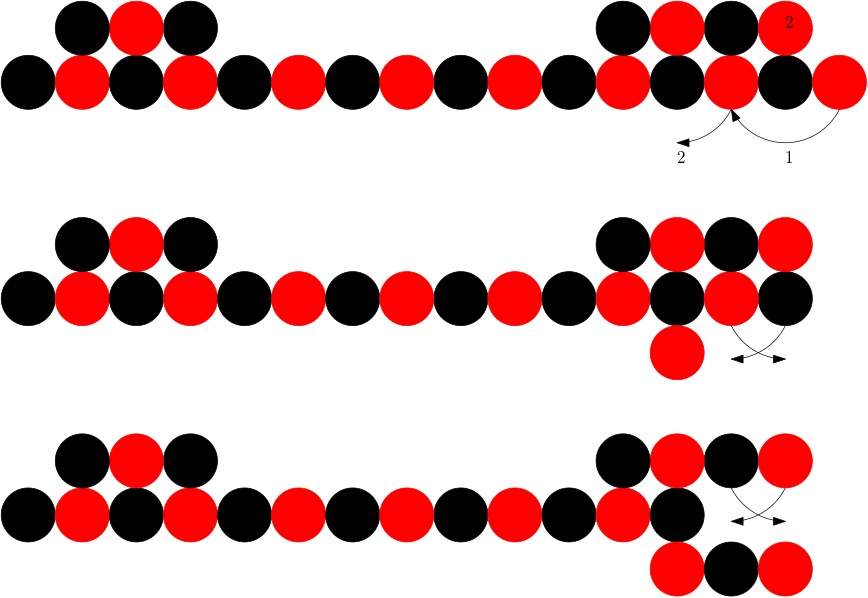}
  \label{fig:sub1}
\end{subfigure}
\begin{subfigure}{.4\textwidth}
  \centering
  \includegraphics[width=.8\linewidth]{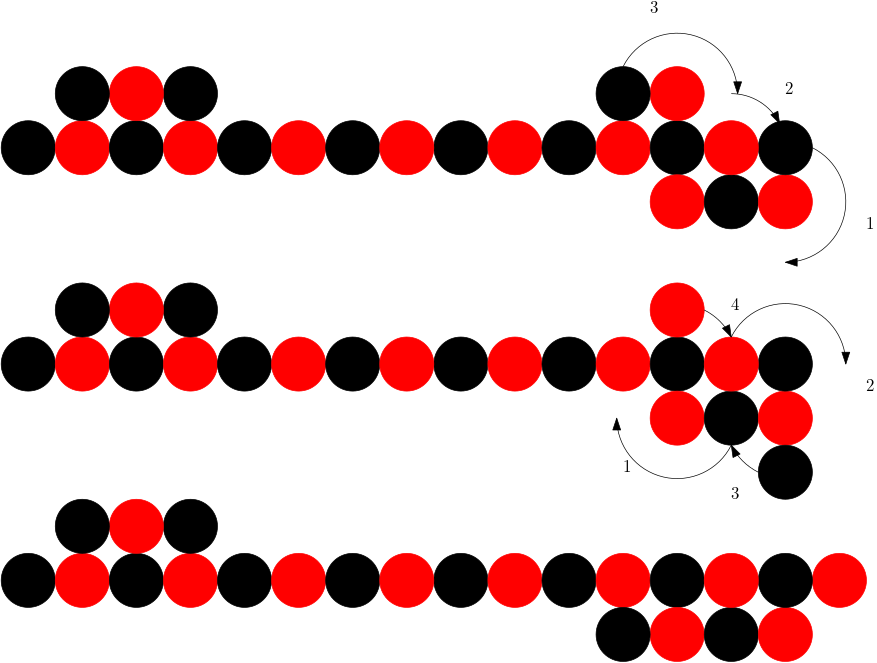}
  \label{fig:sub2}
\end{subfigure}
\caption{Pushing the nodes through the line.}
\label{fig:test}
\end{figure}

\FloatBarrier

\subsection{DepositNode}

Next, we present \emph{DepositNode}, a sub-procedure using the 3 node seed to create a 5 node shape and move a node from the horizontal line to any empty cell which the shape can reach, provided the 5 node shape has the correct colouring, defined as having 3 nodes of the colour which will fill the cell.

We raise two nodes from the line, use this shape to deposit a node and move the other 4 nodes as a square back to the left.
By leaving the cells above and below the two leftmost nodes in the line empty we can rotate the leftmost node, merging with the square to create a new 5 node shape.
We can therefore repeat the process of moving for each node one at a time.
In addition, this sub-procedure can be applied to the builder on the other side of the line.

We provide pseudo-code to describe this process.
If the seed nodes are making their first transfer then they need to raise two nodes to become a builder with 5 nodes.
After that, they only to raise one node at a time.
getLineHead() gets the node which is currently leftmost in the line.
move() moves a node to the destination. It must be the same colour but does not need to be the exact node.

\begin{algorithm}
\caption{DepositNode}\label{euclid}
\begin{algorithmic}[1]
\State $destination \gets (x, y) \text{ //co-ordinates node will be deposited on}$
\If{$seedFirstTransfer == true$ \text{ //If this is the first transfer for the seed nodes}}
 \State $pair \gets \textit{RaiseNodes}$
 \State $node \gets seed[3] \text{ //the node at the front when the seed is formed}$
\Else
 \State $node \gets getLineHead() \text{ //the current leftmost node on the line}$
\EndIf
\State $move(node, destination)$
\end{algorithmic}
\end{algorithm}

Our strategy is to demonstrate that the moves each builder can make are sufficient to be able to construct a nice shape.
We do this by providing examples of the situations which appear when constructing such shapes and proving that the movement we intend to accomplish is possible.
In this example, we show that it is possible to deposit the node at the end of the horizontal line.

\begin{lemma}
A 3 node seed on any line $S$ of length $n$, where $n$ is an even number, can transfer a node the other end of the line.
\end{lemma}

\begin{proof}
We position the 3 node seed in the same grid spaces as in Lemma 4.
We can then follow the process in Figure 5 to achieve our result.
The process of moving right two spaces is repeatable, these repetitions are omitted.
\end{proof}

\begin{figure}
\centering
\begin{subfigure}{.4\textwidth}
  \centering
  \includegraphics[width=.8\linewidth]{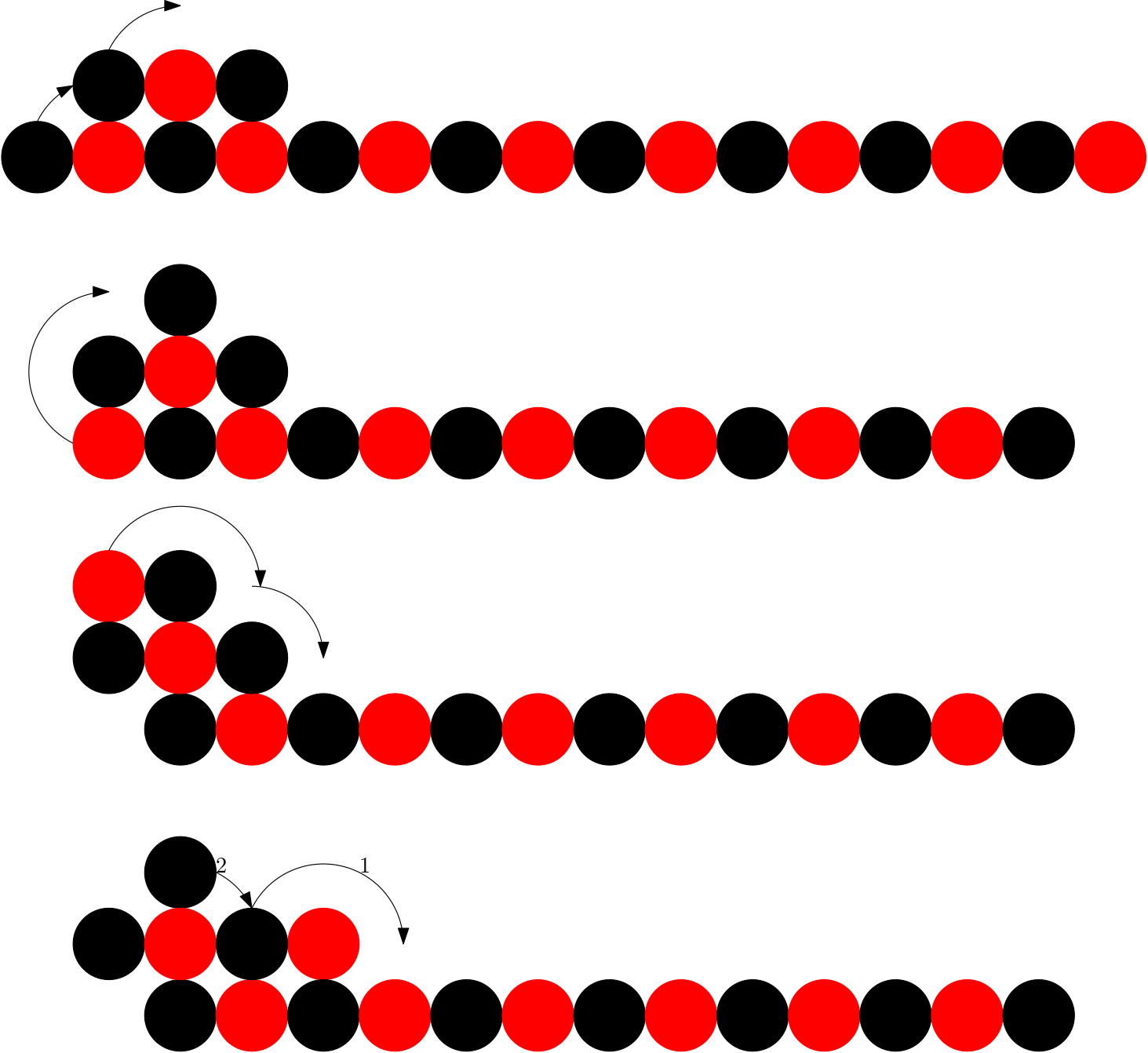}
  \label{fig:sub1}
\end{subfigure}
\begin{subfigure}{.4\textwidth}
  \centering
  \includegraphics[width=.8\linewidth]{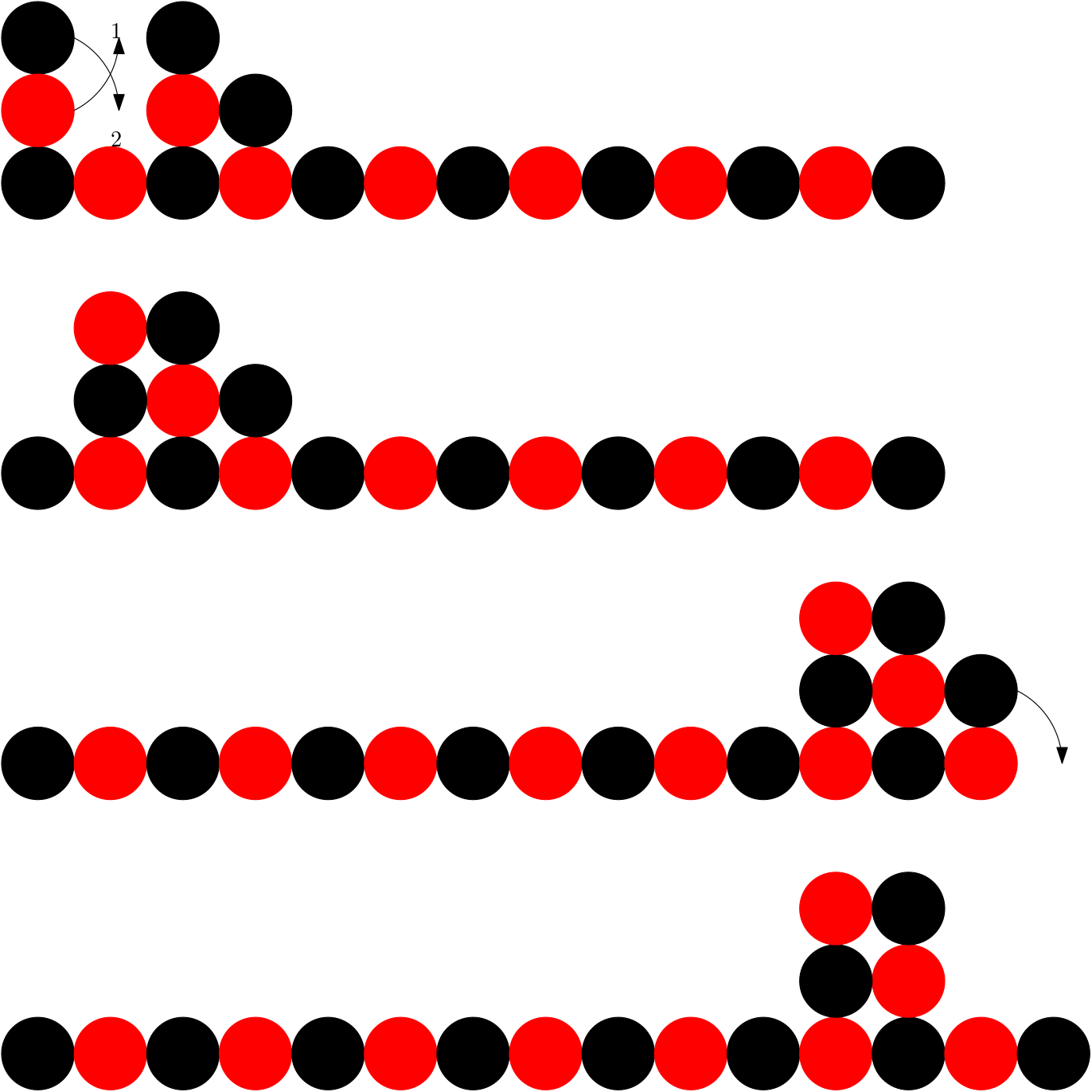}
  \label{fig:sub2}
\end{subfigure}
\caption{Moving a node across the line.}
\label{fig:test}
\end{figure}

\FloatBarrier

If we follow the procedures in Figure 6 it is possible to transfer the builder to a vertical line.
Therefore the process of adding another node can be performed on vertical lines, such as the ones we will build for our nice shapes.

\begin{figure}
\centering
\begin{subfigure}{.4\textwidth}
  \centering
  \includegraphics[width=.8\linewidth]{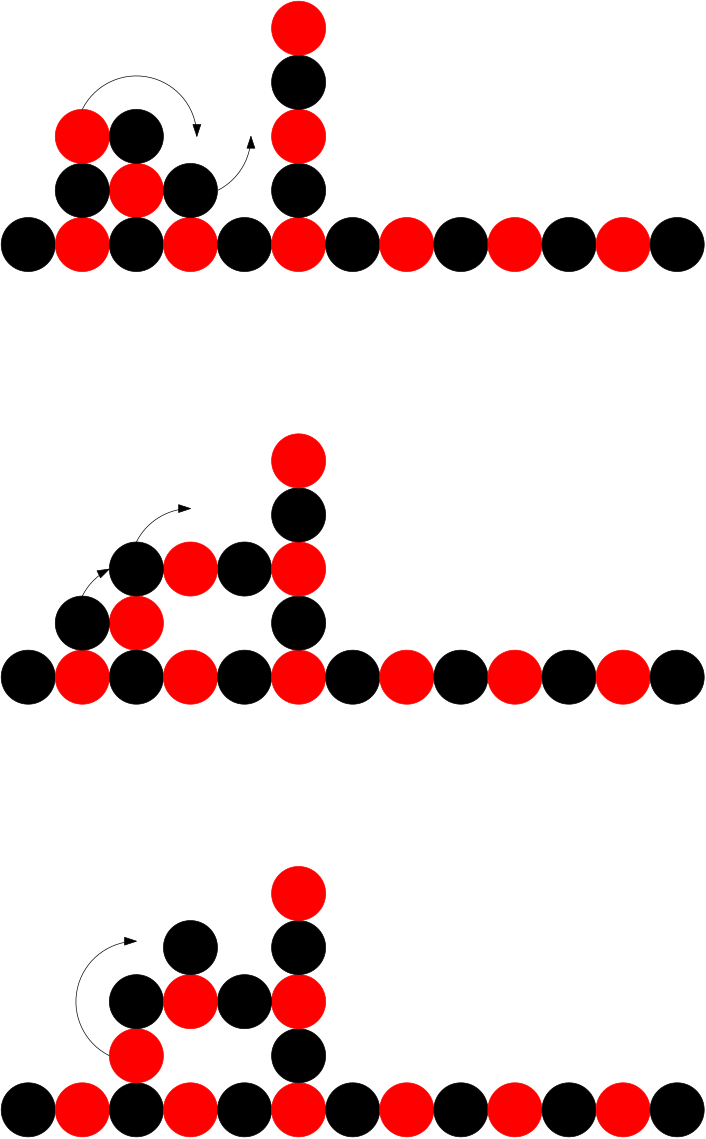}
  \label{fig:sub1}
\end{subfigure}
\begin{subfigure}{.4\textwidth}
  \centering
  \includegraphics[width=.8\linewidth]{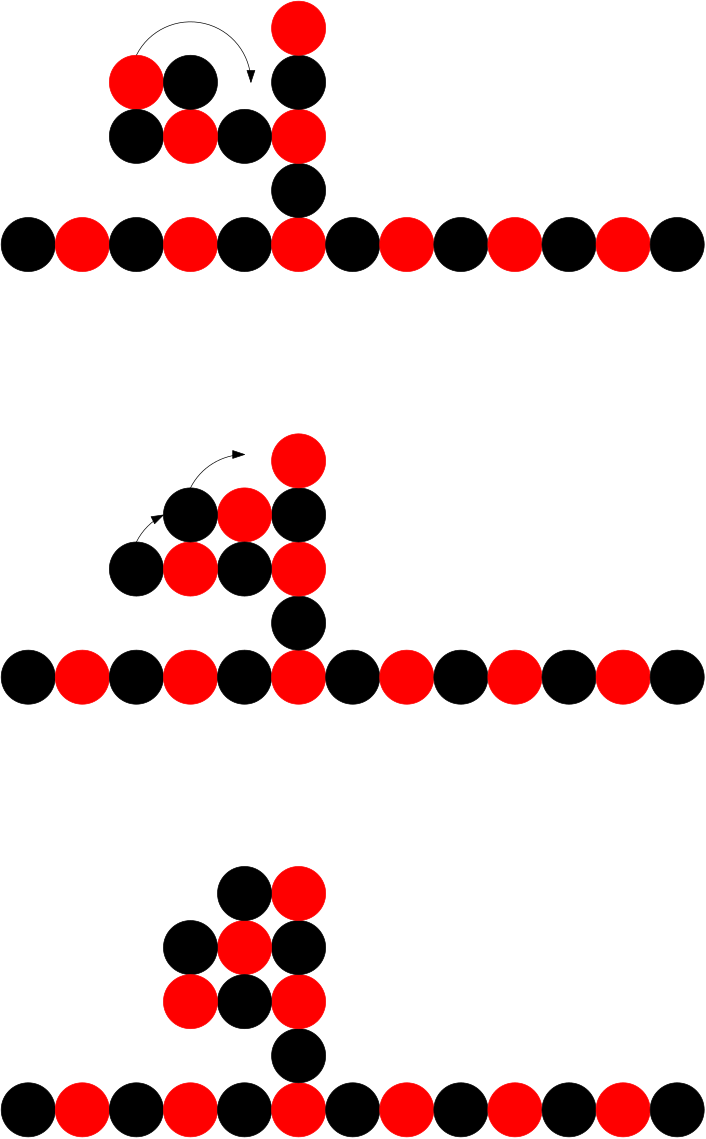}
  \label{fig:sub2}
\end{subfigure}
\begin{subfigure}{.4\textwidth}
  \centering
  \includegraphics[width=.8\linewidth]{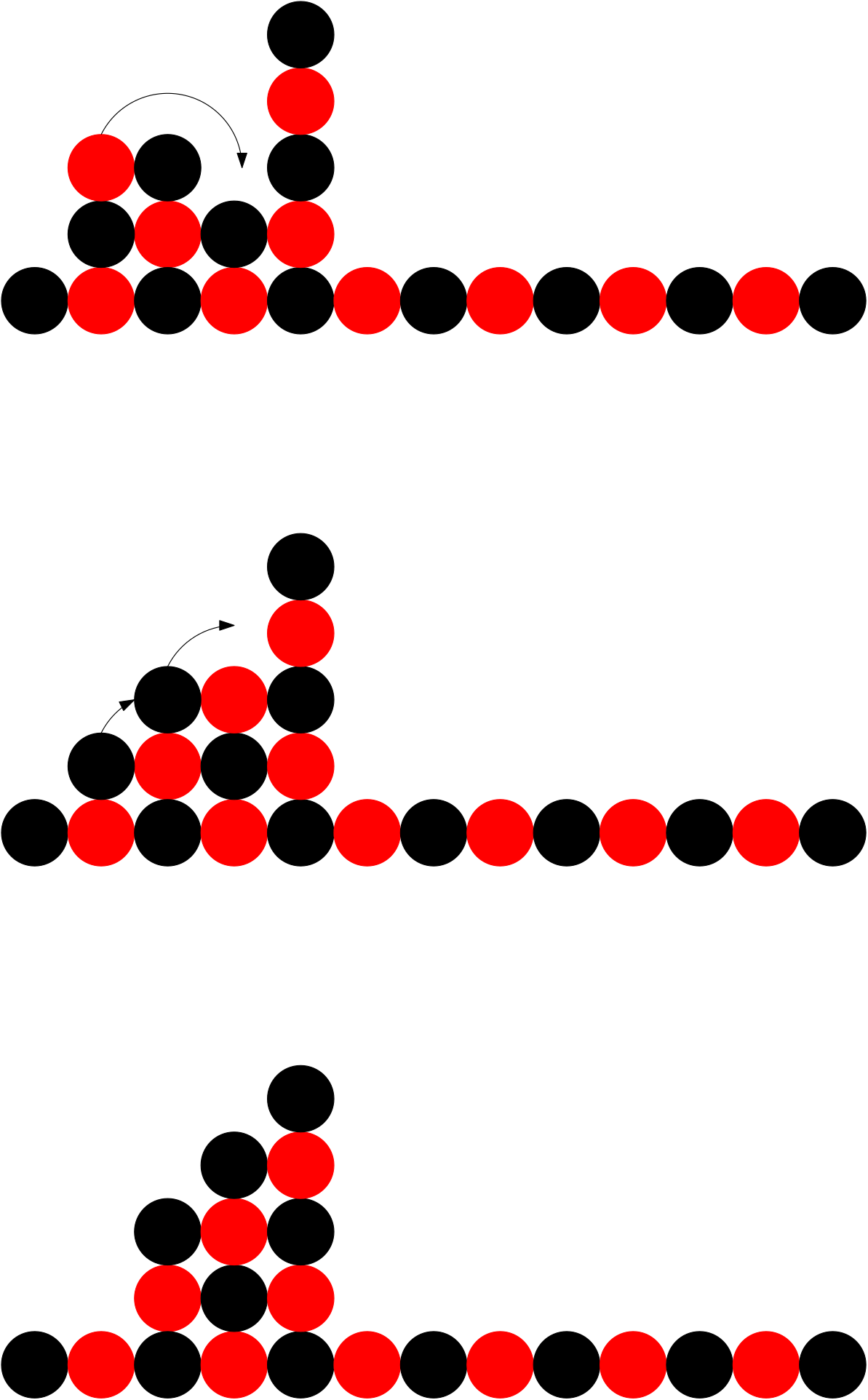}
  \caption{With the line colours reversed.}
  \label{fig:sub2}
\end{subfigure}
\caption{Moving the builder onto a vertical line.}
\label{fig:test}
\end{figure}

\FloatBarrier

To build any vertical line, we must first show that it is possible for DepositNode to construct lines of length 4 above the horizontal line.
After that, because it is possible for the builders to shift onto a 4 node line, the situation becomes that of depositing a node at the end of a line.

\begin{lemma}
Using a 3 node seed in the cells $(1, 1)$ to $(3, 1)$, above a line $L$ of length $n$ it is possible to create another line of length $4$ above any $u_i \in L$.
\end{lemma}

\begin{proof}
For this situation we have two scenarios: one where the colouring is correct and another where it is incorrect.
We first consider the correct colouring and then show how to deal with the incorrect colouring.

For the first node, we simply deposit the node using DepositNode above $u_i$.
The next node is deposited above $u_{i-1}$ and rotated to be above the first.
The next two nodes are more difficult, so we have provided Figure 7 to illustrate the process.

In the case where the colouring is incorrect, we deposit the incorrect node anywhere to the right directly above $L$ and collect a second node from $L$.
We can then merge the 5 node shape with the node we deposited temporarily to create 6 node shape.
Then 5 nodes of the correct colouring can split from the shape we created and deposit the node.

The 4 node square can then return to the node that was left behind and use it as the next node for depositing.
In this way, the 5 node shape is capable of ``selecting'' its colouring.
\qed
\end{proof}

\begin{figure}
\centering
\begin{subfigure}{.4\textwidth}
  \centering
  \includegraphics[width=.8\linewidth]{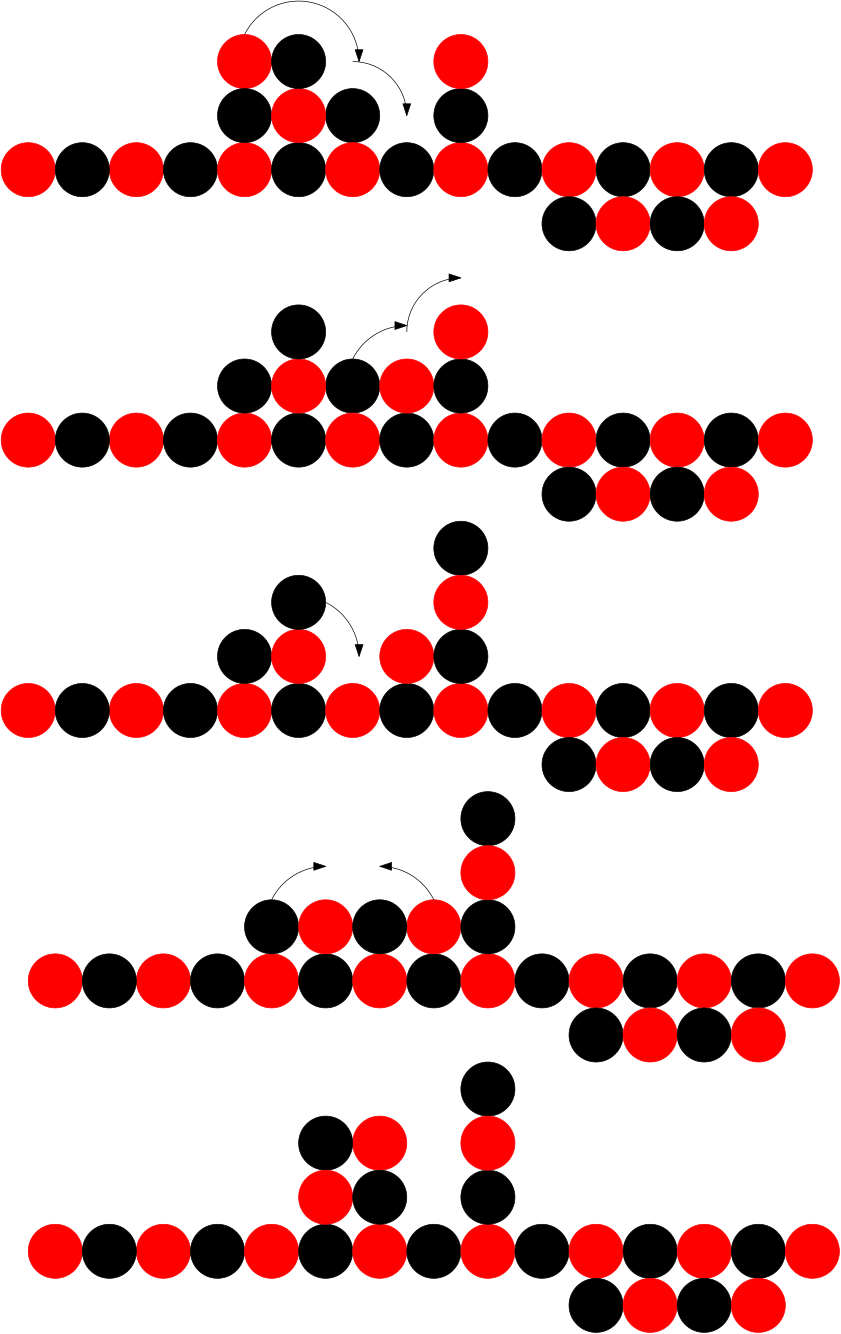}
  \label{fig:sub1}
  \caption{Adding the third node}
\end{subfigure}
\begin{subfigure}{.4\textwidth}
  \centering
  \includegraphics[width=.8\linewidth]{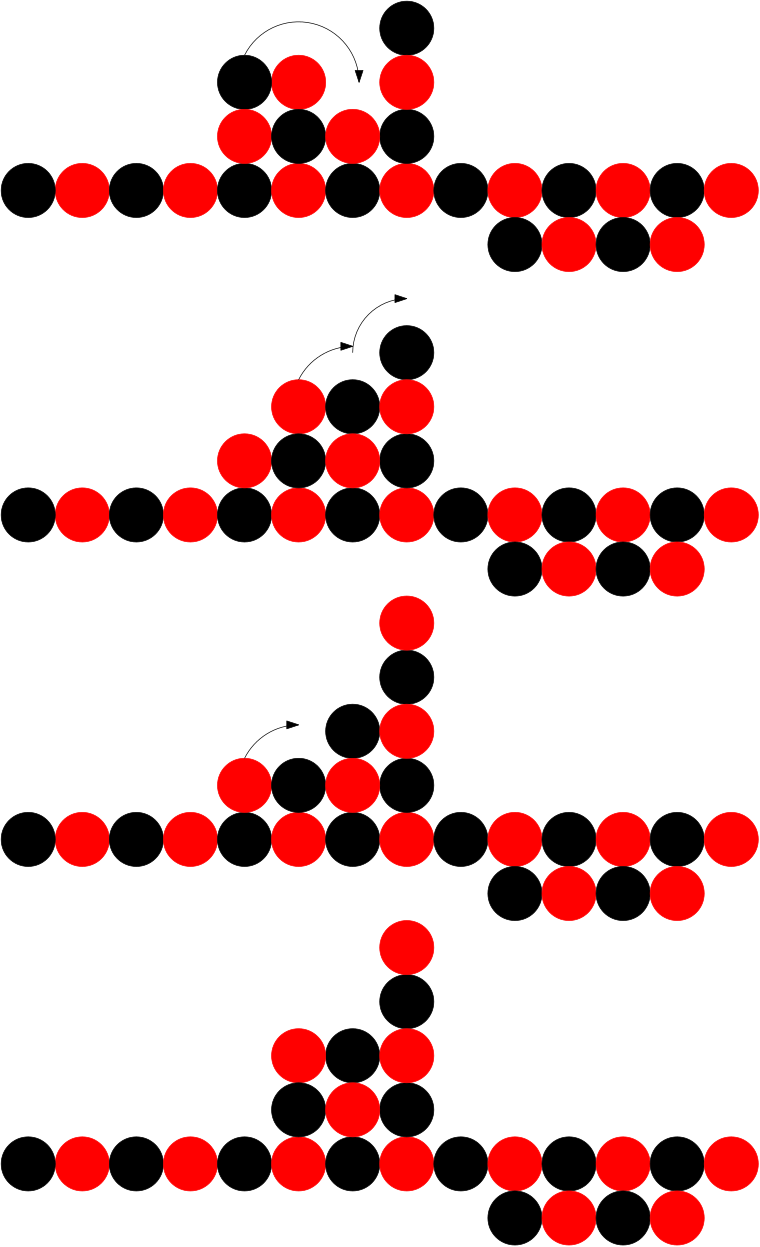}
  \label{fig:sub2}
  \caption{Adding the fourth node}
\end{subfigure}
\caption{Adding the last two nodes.}
\label{fig:test}
\end{figure}

\FloatBarrier

\subsection{Construction of a subset of nice shapes}
We now have all of the lemmas that are necessary to prove that it is possible to construct a specific subset of the class of nice shapes.
We first present an upper bound on the time for constructing nice shapes using our algorithm.
We then prove that using our sub-procedures we can construct a nice shape using a line and a 3-node seed, and finally we show that process is reversible using a 4-node seed.

Let $M_{n-w}$ be the class of nice shapes with $n$ nodes and $w$ waste for which the following property holds:\\

For all lines above and below $L$ the node at the end of each line is the opposite colour to the node at the end of its nearest neighbouring lines. \\

This restriction is necessary to guarantee that no node with the ``wrong'' colour is ever in the position to block construction.

\begin{lemma}
The transformation of a horizontal line of $n$ nodes into a nice shape requires $O(n^2)$ time steps.
\end{lemma}

\begin{proof}
The RaiseNodes and MirrorSeed algorithmic procedures perform a sequence of specific movement and as such are $O(1)$ time.
The DepositNode procedure moves 5 nodes, deposits a node and returns with 4 nodes.
Therefore we must perform $5n + 4n$ moves to transfer one node, and in the worst case, we must build a vertical line of length $n$ above the last node $(n - 1, 0)$ in the horizontal line. This means each node must move past $n - 1$ nodes to reach their destination.

Therefore the process is bound by the speed of DepositNode, which is $(5n + 4n) \cdot (n - 1) = O(n^2)$ time steps.
\qed
\end{proof}

\begin{theorem}
A line of length $n$ can be transformed to any given nice shape in the class $M_{n-1}$ using a 3-node seed in $O(n^2)$ time.
\end{theorem}

\begin{proof}
The seed is positioned above the second, third and fourth nodes in the horizontal line, at $(1, 1)$, $(2, 1)$, and $(3, 1)$.
We first use RaiseNodes twice to raise $4$ nodes from the line and then use MirrorSeed to create a $4$ node builder below the horizontal line.
Then, we use DepositNode to construct the $5$ node builder.

The $5$ node builder can deposit a node in the construction area and move back to the end of the line by having each node rotate around each other.
It is then able to take another node from the horizontal by positioning itself two node spaces away from the end of the line and rotating the last node such that it is connected to the seed.

We are therefore able to follow a procedure for constructing vertical lines one node at a time.
The construction proceeds for each side of $L$ in phases $0 \leq i \leq |L|$, where phase $i$ corresponds to the construction of the column above the node $L_i$.

The entire process is mirrored for the bottom of the shape using DepositNode for the builder on the bottom.
The builder on the bottom waits until the builder on the top is finished and then starts lifting nodes from the same side of the line.
By moving the other builder slightly it is possible to avoid the situation where it disconnects from the line.

Finally, one of the builders places the nodes of the other builder, and is then discarded, leading to a waste of 1 node.
By Lemma 8, the whole process is completed in $O(n^2)$ time.
\qed
\end{proof}

\begin{theorem}
A nice shape in the class $M_n$ can be transformed to any given nice shape from $M_n$ using a 4-node seed in $O(n^2)$ time.
\end{theorem}

\begin{proof}
The transformation can be made reversible by assuming that the $4$ nodes which are discarded at the end of the transformation constitute a $4$ node seed for transforming the nice shape into a line.
We can then construct a line of length $n$ by following the process in reverse, and from there construct a nice shape of size $n$.
\qed
\end{proof}

\subsection{Construction of any nice shape}

We now show how to extend this to the class of all nice shapes.
We follow a broadly similar procedure to the one in Theorem 3.
The key difference is that first create the \emph{foundation}, a layer of nodes above and below the horizontal line.
We place a node at the start of every vertical line which starts with the same node colour that the previous vertical line built would end with.
We then proceed as normal.
First we prove that the foundation is sufficient for constructing any colour-consistent nice shape.
Then we prove that the $5$ node builder is capable of crossing the foundation to deposit nodes.

\begin{lemma}
For any nice shape constructed from a line, for all lines perpendicular to $L$ with an odd number of nodes there is at most one line which cannot be paired with another line which ends in the other colour.
\end{lemma}

\begin{proof}
We have the initial line which is either odd or even.
We can move nodes out in pairs to build lines.
It is possible to build lines which are odd by splitting a pair and distributing its nodes between two odd lines.
Such lines can therefore be paired.

However, there are two ways that an extra odd line can be created.
First, when the horizontal line is odd, we can support one odd vertical line by extracting the extra node.
Second, when the horizontal line is even, we can also split a pair with the horizontal, making it odd.
If both are attempted at the same time the resulting lines will end in different colours and therefore can be paired.
As a result, at most one line which is odd cannot be paired.
\qed
\end{proof}

\begin{lemma}
Any $5$ node builder which is constructing lines can cross the foundation to do so.
\end{lemma}

\begin{proof}
When moving a builder carrying a node across the foundation, there are 3 scenarios the builder can encounter.

In the first scenario, there is a node in $(x, y)$ which is the same colour as the node being carried in $(x - 2, y)$.
In this case, the builder must deposit the node in $(x - 4, y)$ and collect the node it has encountered.
Then, when the builder is returning without carrying a node, it must shift the node it deposited from $(x - 4, y)$ to $(x, y)$.
This process is depicted in figure 8.

In the second scenario, the node at $(x, y)$ is a different colour and the cell $(x + 1, y)$ is empty.
For this scenario, the node which is being carried rotates into $(x + 1, y)$.
Then the builder's nodes rotate around each other to be above $(x - 1, y - 1)$ and $(x - 2, y - 1)$.
Then the top two nodes in $(x - 1, y + 1)$ and $(x - 2, y + 1)$ rotate around each other such that the node in $(x + 1, y)$ is the node being carried by the $5$ node builder.
This process is depicted in figure 9.

In the third scenario, there is a series of nodes beginning with the node $(x, y)$, with alternating colours blocking the builder.
In this case, we first identify the node $n$ which is the node furthest right in the series with the same colour as the node the builder is carrying. Then the top two nodes of the builder in $(x - 2, y + 1)$ and $(x - 3, y + 1)$ rotate until they are positioned such that they form a $5$ node builder with $n$.
This process is depicted in figure 10.

Any foundation must consist of any of these three scenarios arranged in a sequence.
Therefore, by following the correct process in the scenario the builder crosses the foundation and places a node of the correct colour.
Then while returning any nodes deposited can be shifted, creating a new foundation which is equivalent to the original.
\qed
\end{proof}

\begin{figure}
\centering
\begin{subfigure}{.4\textwidth}
  \centering
  \includegraphics[width=.8\linewidth]{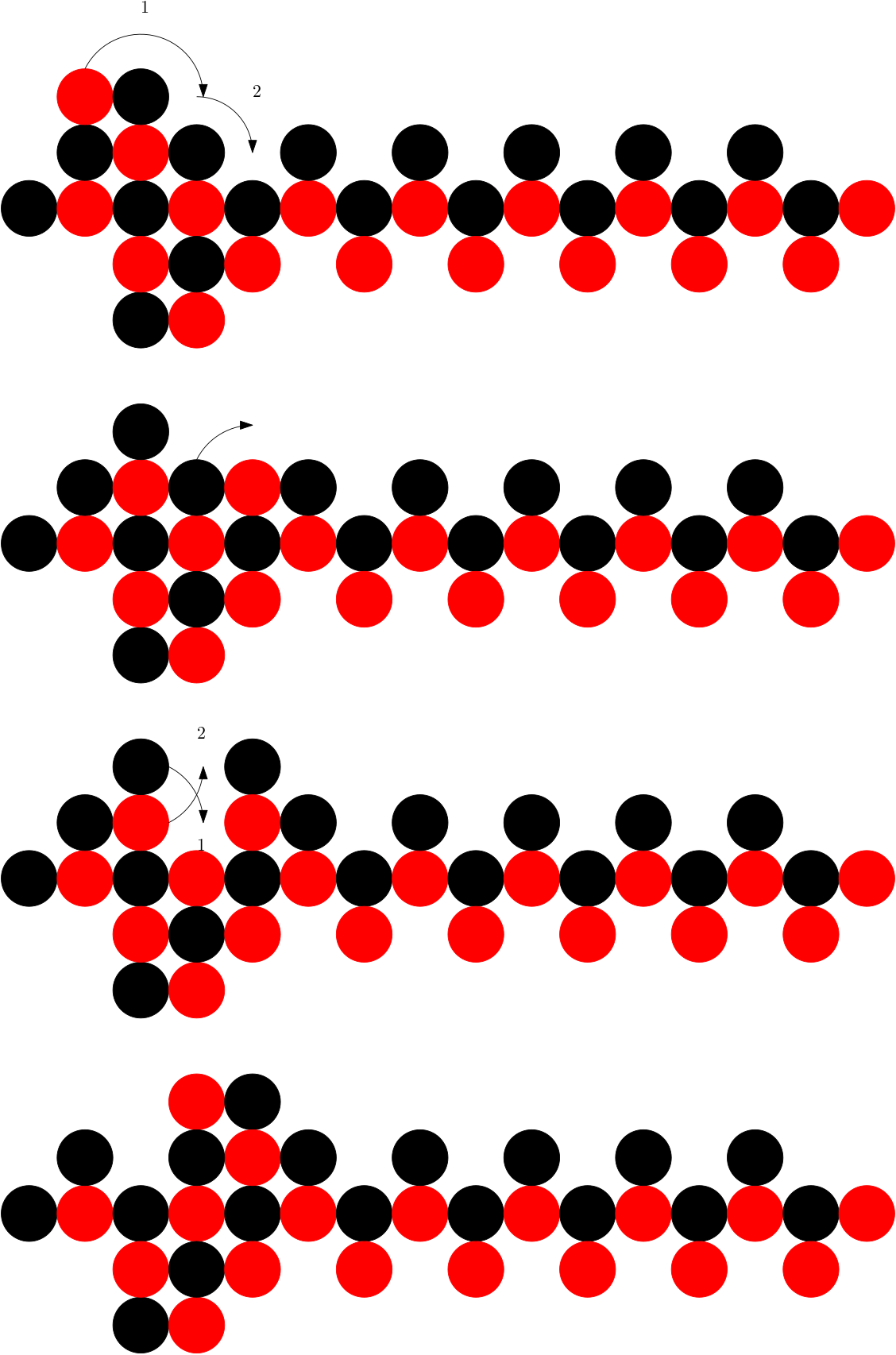}
  \label{fig:sub1}
  \caption{Passing one node of the same colour.}
\end{subfigure}
\begin{subfigure}{.4\textwidth}
  \centering
  \includegraphics[width=.8\linewidth]{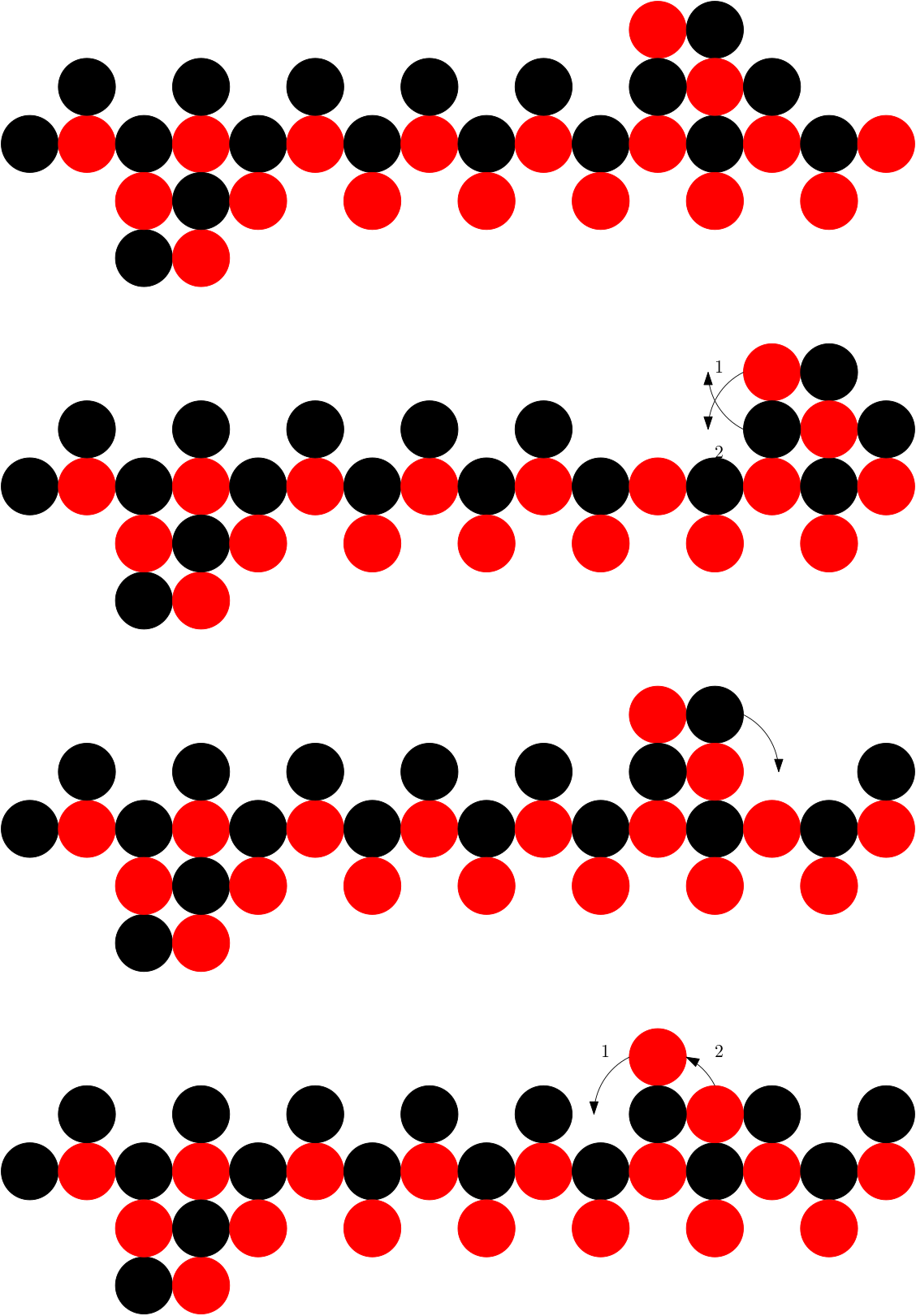}
  \label{fig:sub2}
  \caption{Placing the node.}
\end{subfigure}
\begin{subfigure}{.4\textwidth}
  \centering
  \includegraphics[width=.8\linewidth]{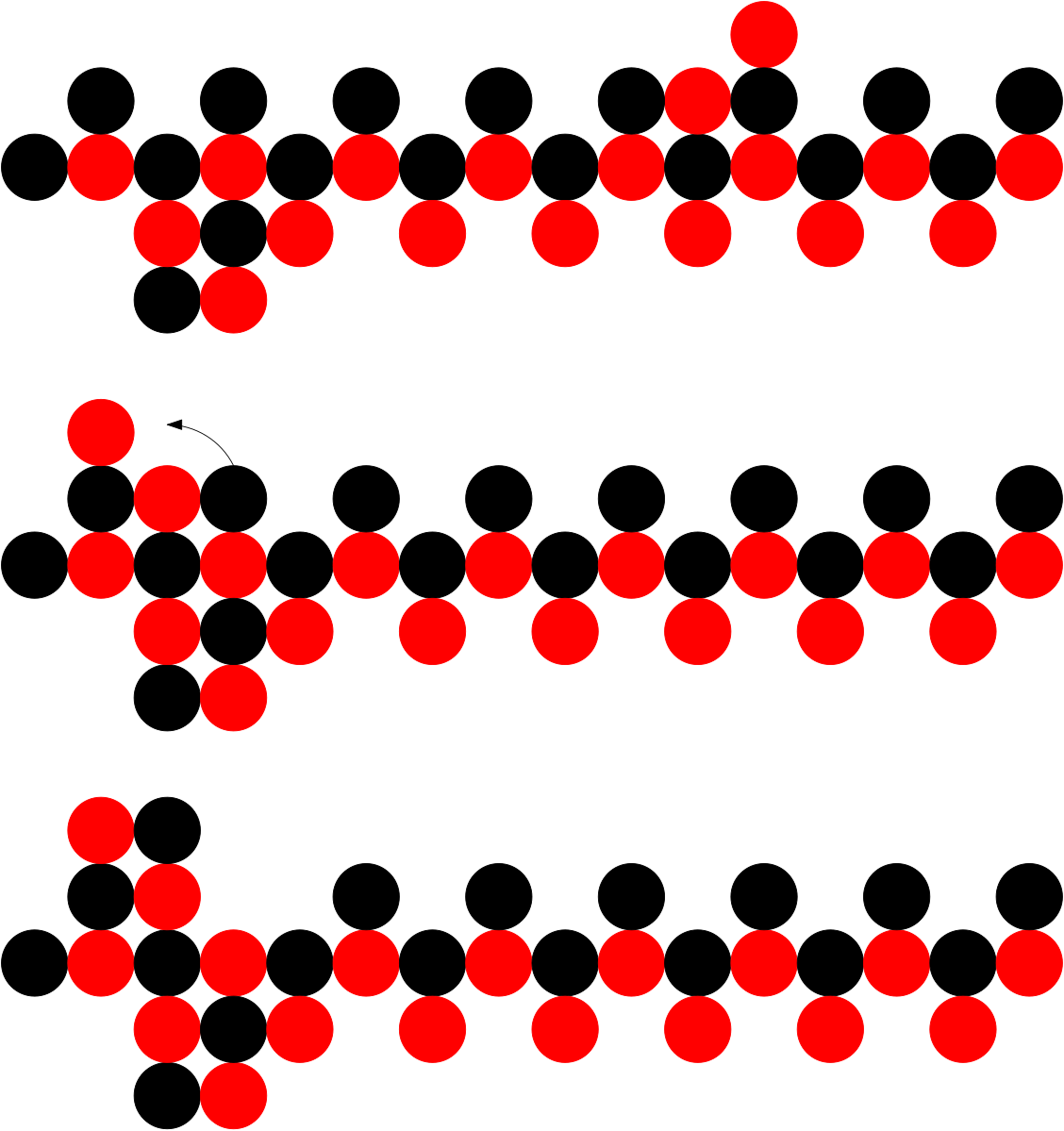}
  \label{fig:sub1}
  \caption{Returning.}
\end{subfigure}
\caption{The case where the colours are the same.}
\label{fig:test}
\end{figure}

\FloatBarrier

\begin{figure}
\centering
\begin{subfigure}{.4\textwidth}
  \centering
  \includegraphics[width=.8\linewidth]{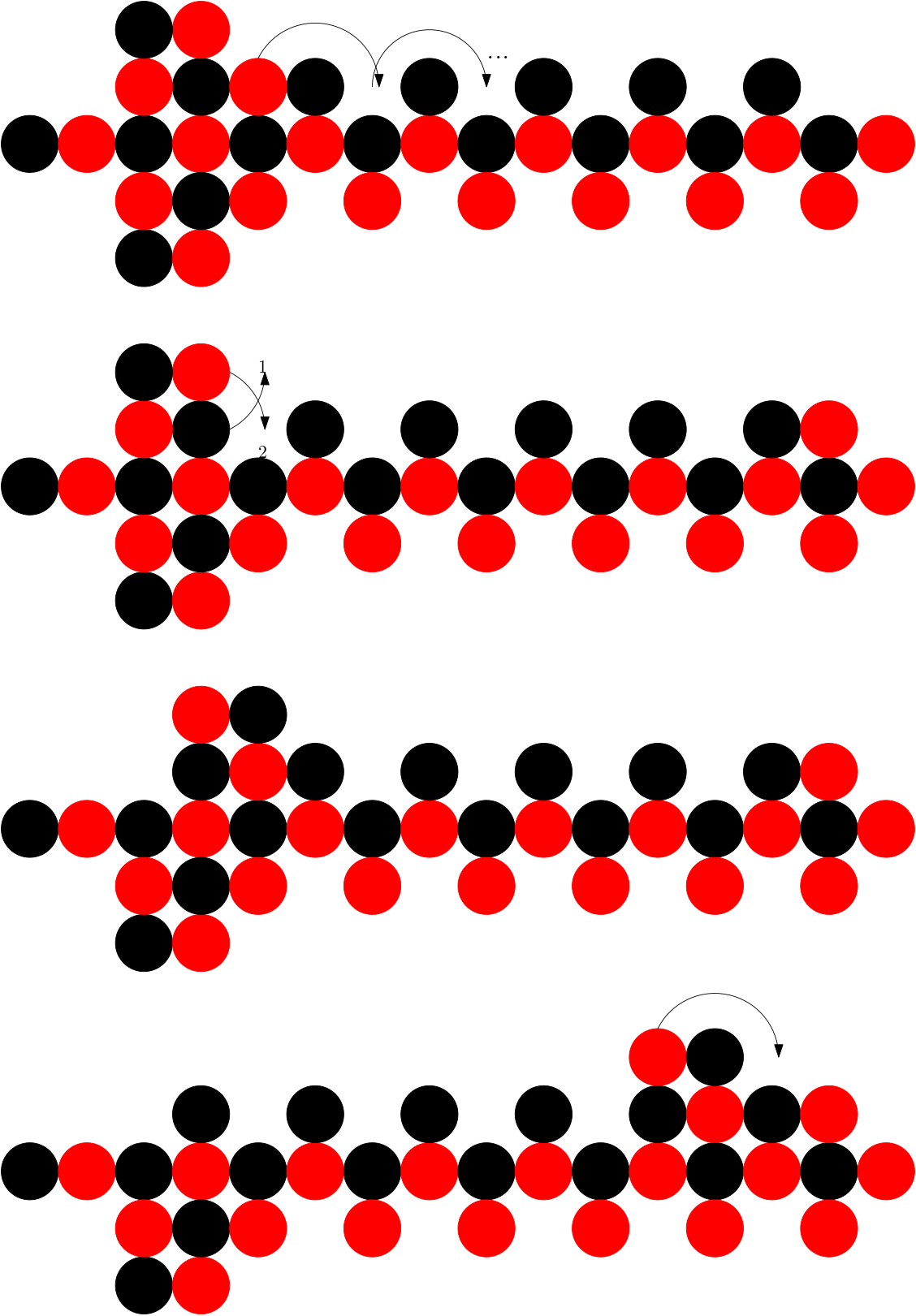}
  \label{fig:sub2}
  \caption{Moving the node, and then the rest of the builder.}
\end{subfigure}
\begin{subfigure}{.4\textwidth}
  \centering
  \includegraphics[width=.8\linewidth]{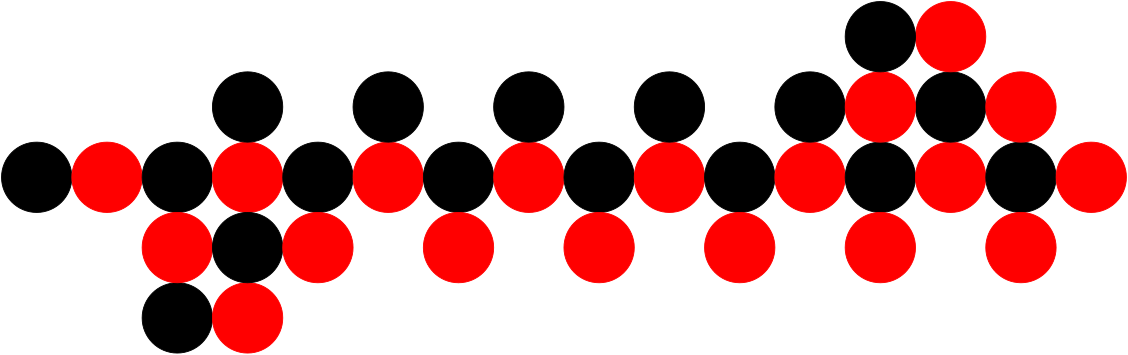}
  \label{fig:sub1}
  \caption{The result.}
\end{subfigure}
\caption{The case where the colours are different.}
\label{fig:test}
\end{figure}

\FloatBarrier

\begin{figure}
\centering
\begin{subfigure}{.5\textwidth}
  \centering
  \includegraphics[width=.8\linewidth]{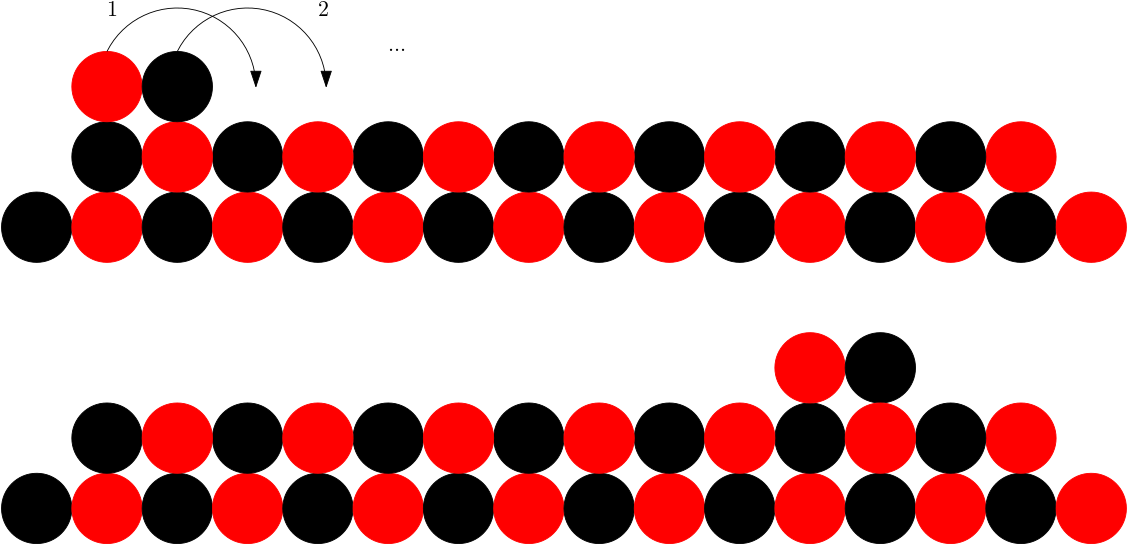}
  \label{fig:sub2}
\end{subfigure}
\caption{The case where the colours alternate.}
\label{fig:test}
\end{figure}

\FloatBarrier

We are now in the position to prove our main result, that it is possible to construct any nice shape from any other nice shape using a seed of size 4.
Let $N_n$ be the subclass of nice shapes which is colour consistent to a line of length $n$.

\begin{theorem}
A line of length $n$ can be transformed to any given nice shape $N_{n-1}$ using a 3-node seed in $O(n^2)$ time.
\end{theorem}

\begin{proof}
The initial steps of the procedure are as in Theorem 3.
When we have created both builders, we then create the foundation by placing each node in the foundation from right to left.
We alternate between the builders as necessary.
By Lemma 9, we know that the scenario where the colours we need to place do not match what is available will never occur.
By Lemma 10, we know that the existence of the foundation does not impede construction.
We are then able to follow a procedure for constructing vertical lines as before.
Finally, the last builder is discarded as before.
\qed
\end{proof}

\begin{theorem}
A nice shape of $n$ nodes can be transformed to any given nice shape $N_n$ using a 4-node seed in $O(n^2)$ time.
\end{theorem}

\begin{proof}
By Theorem 5, we can construct a nice shape from a line using a 3-node seed with 1 node as waste.
By reversibility, we can start with a 4 node seed and construct a line of length $n$.
It is then possible to construct another nice shape using the line.
\qed
\end{proof}

\section{Conclusions}\label{sec6}
Some open problems follow from the findings of our work. The most obvious is expanding the class of shapes which can be constructed using minimal seeds to those which can be derived from nice shapes. This could possibly be expanded by transferring nodes along the perimeter of a nice shape with the help of bridging nodes or by compressing them. In the long run this could lead to characterisations of the classes of connectivity-preserving shapes which can be constructed using only rotation for a given seed. Another important question is the impact that switching to a decentralised model of transformations will have on the results, especially because most programmable matter systems which model real-world applications implement programs in this way. This in turn could lead to real-world applications for the efficient transformation of programmable matter systems.

\bibliographystyle{ieeetr}
\bibliography{newpaper}

\end{document}